\documentclass[preprint,12pt]{elsarticle}

\usepackage{amsmath,amssymb,amsfonts,graphics,
latexsym, amscd, amsfonts, epsfig, eepic,epic,psfrag,setspace,eufrak}
\usepackage[all]{xypic}
\usepackage{color}

\usepackage{amssymb}
\usepackage{amsthm}

\theoremstyle{plain}

\newtheorem{lemma}{Lemma}
\newtheorem{proposition}{Proposition}
\theoremstyle{definition}
\newtheorem{definition}{Definition}

\newtheorem{remark}{Remark}
\theoremstyle{remark}

\journal{Mathematical Biosciences}

\begin{document}

\begin{frontmatter}

%% Title, authors and addresses

%% use the tnoteref command within \title for footnotes;
%% use the tnotetext command for theassociated footnote;
%% use the fnref command within \author or \address for footnotes;
%% use the fntext command for theassociated footnote;
%% use the corref command within \author for corresponding author footnotes;
%% use the cortext command for theassociated footnote;
%% use the ead command for the email address,
%% and the form \ead[url] for the home page:
%% \title{Title\tnoteref{label1}}
%% \tnotetext[label1]{}
%% \author{Name\corref{cor1}\fnref{label2}}
%% \ead{email address}
%% \ead[url]{home page}
%% \fntext[label2]{}
%% \cortext[cor1]{}
%% \address{Address\fnref{label3}}
%% \fntext[label3]{}

\title{Topological language for RNA}

%% use optional labels to link authors explicitly to addresses:
%% \author[label1,label2]{}
%% \address[label1]{}
%% \address[label2]{}

\author[vt]{Fenix W.D.\ Huang}
\ead{fenixprotoss@gmail.com}

\author[vt]{Christian M.\ Reidys\corref{cor1}}
\ead{duck@santafe.edu}

\cortext[cor1]{Corresponding author}

\address[vt]{Biocomplexity Institute of Virginia Tech, Virginia Tech,
United States of America}

\begin{abstract}
%% Text of abstract
In this paper we introduce a novel, context-free grammar, {\it RNAFeatures$^*$},
capable of generating any RNA structure including pseudoknot structures (pk-structure).
We represent pk-structures as orientable fatgraphs, which naturally leads to a
filtration by their topological genus. Within this framework, RNA secondary structures
correspond to pk-structures of genus zero.
{\it RNAFeatures$^*$} acts on formal, arc-labeled RNA secondary structures, called
$\lambda$-structures. $\lambda$-structures correspond one-to-one to pk-structures
together with some additional information. This information consists of the
specific rearrangement of the backbone, by which a pk-structure can be made
cross-free.
{\it RNAFeatures$^*$} is an extension of the grammar for secondary
structures and employs an enhancement by labelings of the symbols as
well as the production rules. We discuss how to use {\it RNAFeatures$^*$} to
obtain a stochastic context-free grammar for pk-structures, using data of RNA
sequences and structures. The induced grammar facilitates fast Boltzmann sampling
and statistical analysis. As a first application, we present an $O(n log(n))$ runtime
algorithm which samples pk-structures based on ninety tRNA sequences and structures
from the Nucleic Acid Database (NDB).
\end{abstract}

\begin{keyword}
%% keywords here, in the form: keyword \sep keyword

%% PACS codes here, in the form: \PACS code \sep code

%% MSC codes here, in the form: \MSC code \sep code
%% or \MSC[2008] code \sep code (2000 is the default)
RNA pseudoknot structure \sep context-free grammar \sep topological RNA structure \sep
fatgraph \sep stochastic context-free grammar
\end{keyword}

\end{frontmatter}

%% \linenumbers

%% main text
\section{Introduction} \label{S:Int}

RNA secondary structure provides a coarse grained model to study important
features of the real molecular conformation, embedded in three-space \cite{Uhlenbeck:98}. The key
feature of these structures is that they can be inductively constructed \cite{Stein:78a}.
The recursion forms the basis for more than three decades of research resulting in
what can be called the dynamic programming (DP) paradigm. DP methodology facilitates
{\it ab initio} folding \cite{Waterman:86, Zuker:81,Hofacker:94a}, partition function 
\cite{McCaskill:90}, Boltzmann sampling \cite{SecondaryStructureSampling} 
and context-free grammars (CFGs) \cite{Harrison:78}.
Minimum free energy (mfe) secondary structures can be folded in $O(n^2)$ space and
$O(N^3)$ time complexity \cite{Waterman:86} and their partition function is computed in $O(n^3)$ time
\cite{McCaskill:90}. Stochastic folding is based on their CFG, allowing to factor in specific
properties of biological structures. Thermodynamic and stochastic secondary structures have
been statistically analyzed in \cite{Nebel:11, Fontana:93, Hofacker:94a}.
In \cite{Rivas:99} the secondary structure paradigm has been extended to include gap-structures.
Gap-structures include certain, but not all pseudoknot structures (pk-structures). They can be
viewed as the most general class to which the DP-paradigm directly applies and their underlying
grammar is multiple context-free. The approach allowed folding of gap-structures in $O(n^6)$
time and $O(n^4)$ space complexity and sampling in $O(n^3)$ time.

In this paper we approach the problem of cross-serial interactions of pk-structures from first
principles. We employ an abstraction, that is well known to structural biology \cite{Leontis:01}: namely,
a nucleotide is not a dimensionless object but a triangle and base pairs can be abstracted as
untwisted or twisted ribbons connecting two such triangles, see Fig.~\ref{F:basepair}.
This observation introduces topology as a natural language for biological structures and leads to
the fatgraph model \cite{Penner:10}.
Our main result is to map a novel recursion on unicellular maps \cite{Chapuy:11} into the classic
recursion derived in \cite{Stein:78a}. I.e.~we translate a recursive procedure on fatgraphs into
an enhanced version of the recursion of secondary structures. This results in the first,
unambiguous, context-free grammar for pk-structures, {\it RNAFeatures$^*$}. We present as a
first application an $O(n \log(n))$ Boltzmann sampler and stochastic context-free grammar.

In the following we discuss the fatgraph modeling Ansatz in detail. We then provide some
background on computational linguistics in the context of RNA structures, enabling us to
discuss what we mean by topological language of RNA. 

The biophysics of the canonical base pairs {\bf A-U}, {\bf G-C} and {\bf G-U} implies, that
the fatgraphs modeling RNA structures are orientable: RNA purine and pyrimidine bases are
modeled as three edges for hydrogen bonding interactions: the Watson-Crick edge, the Hoogsteen
edge, and the Sugar edge\footnote{which includes the 29-OH and which has also been referred to as the
Shallowgroove edge} \cite{Leontis:01}. 
A given edge of one base can potentially interact in a plane with any one of the three
edges of a second base, and can do so in either cis- or trans- orientation of the glycosidic bonds.
This gives rise to 12 basic geometric types with at least two H bonds connecting the bases,
see Fig.~\ref{F:basepair}. Accordingly, cis-base pairs, including the 
canonical base pairs belong to the {\it cis-Watson-Crick/Watson-Crick} geometry.

Canonical base pairs induce exclusively untwisted ribbons, modeled via orientable \emph{fatgraphs},
i.e.~graphs enhanced by replacing vertices by discs and edges by ribbons. Gluing the sides of these
ribbons produces the topological quotient space \cite{Massey:69}, which is a connected sum of tori.
As topological genus, $g$, and orientability are the two determinants of closed surfaces, RNA structures
having canonical base pairs are filtered by genus alone. RNA secondary structures \cite{Waterman:78aa},
are exactly structures of genus zero. Structures of genus $g>0$ are those exhibiting cross-serial
interactions. In a sense, topological genus measures the complexity of a biomolecule, far more subtle
than the number of crossing arcs.
One may think of pk-structures as {\it cross-free} drawings on a closed surface of genus $g\ge 0$,
instead of a drawing with crossing arcs in the plane\footnote{A cross-free drawing in the plane
corresponds to an RNA secondary structure.}, see Fig.~\ref{F:pse}. 

%%%
%%%%%%%%%%%%%%%%%%%%%%%%%%%%%%%%%%%%%%%%%%%%%%%%%%%%%%%%%%%%%%%%%%%%%%
%%%
%\begin{figure}[ht]
%\centerline{\includegraphics[width=\columnwidth]{PK_on_tori.eps}}
%\caption{
%  Crossings depend on the surface the structure is drawn on:
%  here the diagram exhibits crossing arcs when drawn in the plane,
%  on the torus however, the diagram can be drawn cross-free. }
%\label{F:pse}
%\end{figure}
%%%
%%%%%%%%%%%%%%%%%%%%%%%%%%%%%%%%%%%%%%%%%%%%%%%%%%%%%%%%%%%%%%%%%%%%%%%
%%%

The term \emph{language of RNA} has been used in \cite{Rivas:00}, where the authors
study the aforementioned gap-structures \cite{Rivas:99}.
Methods from computational linguistics have been applied to problems
in DNA and RNA sequence analysis for decades: early work involved using regular grammars
and hidden Markov models (HMMs) in order to model biological sequences. 
These have been used in sequence analysis \cite{Pachte:05}, such as identifying CpG
islands \cite{Durbin:98}, gene prediction \cite{Munch:06}, pairwise and multiple sequence
alignment \cite{Durbin:98,Pachter:02}, modeling DNA sequencing errors \cite{Lottaz:03},
protein secondary structure prediction, and RNA structural alignment \cite{Yoon:08}.
Regular grammars and HMMs are not well-equipped to model problems in RNA folding
because they cannot describe the long-distance correlations of base pairs. Instead, a
larger class called \emph{context-free} grammars (CFG) are used.
Assigning probabilities to the production rules of CFG produces a stochastic 
context-free grammar (SCFG), which, in the context of HMM, allows each state to
generate not a single, but any (fixed) number of immediate successors.

By a \emph{topological language of RNA} we mean an unambiguous, CFG that allows us to
recursively construct {\it any} pk-structure using recursions that are induced by cell-surgery
on fatgraphs. The latter shall be translated into an enhancement of the classical recursion \cite{Stein:78a}.
This allows us to connect to the existing paradigm by presenting a CFG based on arc-removals and
decomposition as production rules\footnote{as for  CFG of secondary structures}, 
where the cross-serial interactions are encapsulated in the labels of the arcs and nonterminal symbols,
respectively.

%%%
%%%%%%%%%%%%%%%%%%%%%%%%%%%%%%%%%%%%%%%%%%%%%%%%%%%%%%%%%%%%%%%%%%%%%%
%%%
%General Background 
%%%
%%%%%%%%%%%%%%%%%%%%%%%%%%%%%%%%%%%%%%%%%%%%%%%%%%%%%%%%%%%%%%%%%%%%%%
%%%
Finally we discuss the above mentioned recursion of secondary structures and provide some background
on topological RNA structures.

Waterman {\it et al.} \cite{Waterman:78a, Stein:78a,Nussinov:78,Kleitman:70}
studied the combinatorics and folding of RNA secondary structures. Their diagrams are
labeled graphs over the vertex set $[n]=\{1, \dots, n\}$, presented by drawing the
vertices on a horizontal line and noncrossing arcs in the upper half-plane.
Vertices and arcs correspond to the nucleotides {\bf A}, {\bf G}, {\bf U} and {\bf C}
and Watson-Crick ({\bf A-U}, {\bf G-C}) and wobble ({\bf U-G}) base pairs, respectively.
The noncrossing arcs of RNA secondary structures allow for a recursive build.
Let $S_2(n)$ denotes the number of RNA secondary structures over $n$ nucleotides
then we have \cite{Waterman:78a}: $S_2(n)=S_2(n-1)+
\sum_{j=0}^{n-3}S_2(n-2-j)S_2(j)$, where $S_2(n)=1$ for $0\le n\le 2$. Accordingly,
RNA secondary structures satisfy a constructive recursion. This relation constitutes the basis
for the DP-recursions used for the polynomial time folding \cite{Waterman:78aa, Zuker:81} and
has profound algorithmic implications. The DP-framework has ever since strongly influenced the field
of RNA folding \cite{Waterman:78a, Zuker:81, Sankoff:92, Hofacker:94a}. 

CFGs of secondary structures \cite{Zuker:81,Hofacker:94a, Knudsen:99} are utilizing the recursive
nature of secondary structures \cite{Waterman:78a}. The grammar has two production rules, displayed in
Fig.~\ref{F:grammar}. Its terminal symbols denote a pair of vertices forming an arc as well
as an unpaired vertex and the nonterminal symbols are an arbitrary secondary structure and 
an irreducible structure covered by an external arc.
%\begin{figure}[ht]
%\begin{center}
%\includegraphics[width=0.9\columnwidth]{grammar}
%\end{center}
%\caption{Illustration of the two production rules of the CFG of RNA secondary structures:
%  decomposition (left) and arc removal (right).
%}
%\label{F:grammar}
%\end{figure}
The grammar is unambiguous, i.e.~each structure has a unique decomposition path.

Cross-serial interactions or pseudoknots have long been known as important structural
elements \cite{Westhof:92a, Chen:05pnas} in RNA. Depicting a contact structure as a diagram,
%where the sequence is placed on the $x$-axis and the bonds are drawn as arcs in the upper
%half-plane,
a structure contains a pseudoknot if and only if these arcs cross. Cross-serial interactions
are functionally important in tRNAs, RNaseP \cite{Loria:96a}, telomerase RNA \cite{Staple:05},
and ribosomal RNAs \cite{Konings:95a}. Pseudoknots in plant virus RNAs mimic tRNA structures,
and {\it in vitro} selection experiments have produced pseudoknotted RNA families that bind
to the HIV-1 reverse transcriptase \cite{Tuerk:92}. 

Topological RNA structures have been introduced in \cite{Waterman:93,Penner:03} and
the classification and expansion of pk-structures in terms
of the topological genus of an associated fatgraph has been studied by means of
matrix theory in \cite{Orland:02,Bon:08}. 
The computation of the genus of a fatgraph is classic \cite{Euler:1752} and was
first applied to RNA structures by \cite{Orland:02} and \cite{Bon:08}.
\cite{Reidys:top1} studies topological RNA structures of higher genus and connects them
with Riemann's Moduli space \cite{Penner:03}. In \cite{Reidys:11a} a polynomial time,
loop-based folding algorithm of topological RNA structures was given.

The paper is organized as follows: in Section~\ref{S:fat} we discuss the idea behind
the novel recursions. Thereafter we establish the correspondence between pk-structures and blueprints
with $\lambda$-structures in Section~\ref{S:bi}. 
Then we discuss {\it RNAFeatures$^*$} and finally detail the SCFG in Section~\ref{S:grammar}.

%%%%%%%%%%%%
%%%%%%%%%%%%%%%%%%%%%%%%%%%%%%%%%%%%%%%%%%%%%%%%
%%%%%%%%%%%%

\section{Fatgraphs and blueprints.}\label{S:fat}

The contacts of nucleotides within an RNA structure can be represented as a diagram
whose arcs can be linearly ordered via their left end-points, i.e.~$(i,j)\prec (r,s)$
if and only if $i<r$. Suppose $(i,j) \prec (r,s)$, then $(i,j)$ and $(r,s)$ cross iff 
$i<r<j<s$ holds. 

In the following we shall assume that any diagram contains the arc $(0,n+1)$ (rainbow),
which is not accounted for as an edge of the diagram. Rainbows mark the start and endpoint
of the $5'$-$3'$ oriented backbone, thus allowing to collapse the backbone into a disc without loosing any information. 

The passage from diagrams to fatgraphs \cite{Loebl:08,Penner:10} is obtained by
``thickening'' the edges into (untwisted) bands or ribbons. Furthermore, each vertex
is inflated into a disc as shown in Fig.~\ref{F:fat}. This inflation reflects the fact
that canonical base pairs fix a plane, see Fig.~\ref{F:basepair}.
A fatgraph, $\mathbb{G}$, can thus be viewed as a ``drawing'' on an orientable surface
$X_{\mathbb{G}}$, which is obtained by identifying the sides of the ribbons. $\mathbb{G}$ is
a $2$-dimensional cell complex over its geometric realization, $X_{\mathbb{G}}$.

%%%%%%%%%%FIGURE diagram, fatgraph %%%%%%%%%
%%%%%%%%%%%%%%%%%%%%%%%%%%%%%%%%%%%%%%%%%%%
%%%%%%%%%%%%%%%%%%%%%%%%%%%%%%%%%

%\begin{figure}[ht]
%\begin{center}
%  \includegraphics[width=\columnwidth]{fat.eps}
%\end{center}
%\caption{(A): Inflation of edges and vertices into ribbons and disks.
%  Here we have four boundary components traversed in counter-clockwise
%  orientation traversing the sides of any ribbon in opposite
%  directions. (B): Collapsing the backbone of (A) into a single disc does
%  not affect genus. Here we have $g=1$.
%} \label{F:fat}
%\end{figure}

%%%%%%%%%%%%%%%%%%%%%%%%%%%%%%%
%%%%%%%%%%%%%%%%%%%%%%%%%%%%%%%%%%%%%%%%%%%
%%%%%%%%%%%%%%%%%%%%%%%%%%%%%%%
 
The Euler characteristic and topological genus 
of the surface $\mathbb{D}$ are given by  
$\chi(\mathbb{D})  =  v - e + r$ and $g(\mathbb{D}) =  1-\frac{1}{2}\chi(\mathbb{D})$,
where $v$, $e$ and $r$ denote the number of discs, ribbons and boundary components
in $\mathbb{D}$ \cite{Massey:69}. 
The equivalence of simplicial and singular homology \cite{Hatcher:02} implies that
these combinatorial invariants are topological invariants. This means the genus of
the surface $X_{\mathbb{D}}$ provides a filtration of fatgraphs. 

A fatgraph can be presented by a pair of permutations \cite{Penner:10}. 
Let $H= [2n]$ denote the set of half-edges, and let $\sigma$, $\alpha$ and $\gamma$ be 
three permutations over $H$, where each cycle in $\sigma$, $\alpha$ and $\gamma$
presents a vertex, an edge and a boundary component, respectively.
A vertex of degree $k$ is considered as a cycle $v=(i,\sigma(i),\dots,\sigma^{k-1}(i))$.
In particular, $\alpha$ is a fixed-point free involution since an edge consist of two 
half-edges and we have $\gamma = \alpha \circ \sigma$.

Removing the unpaired vertices from the diagram and collapsing the backbone into a disc
does not change the Euler characteristic. Therefore, the relation between genus and number
of boundary components is solely determined by the number of arcs in the upper half-plane:
$2-2g-r = 1-n$, where $n$ is number of arcs and $r$ the number of boundary components. The latter
can be computed easily and allows us to obtain the genus of the diagram.

Mapping $(H,\sigma, \alpha)$ into $\pi((H,\sigma, \alpha))=(H, \alpha\circ \sigma,\alpha)$,
is a bijection, see Fig.~\ref{F:dual}. $\pi$ is called the {\it Poincar\'{e} dual} and
interchanges boundary components with vertices, preserving topological genus. By construction
we shall deal only with fatgraphs over one backbone whose duals thus have one boundary
component. We refer to these duals as {\it unicellular maps}.

%%%
%%%%%%%%%%%%%%%%%%%%%%%%%%%%%%%%%%%%%%%%%%%%%%%%%%%%%%%%%%%%%%%%%%%%%%%%%%%%%%%%%%%%%%%%%%%%%
%%%
%\begin{figure}[ht]
%\begin{center}
%  \includegraphics[width=\columnwidth]{dual.eps}
%\end{center}
%\caption{(A): a fatgraph represented by permutations. We have $\sigma=(1,2,3,4,5,,6,7,8,9,10)$,
%  the unique vertex obtained by collapsing the backbone, $\alpha = (1,10)(2,5)(3,8)(4,7)(6,9)$ the
%  fixed-point free involution representing the ribbons, and $\gamma=(1,5,9)(2,8,6,4)(3,7)(10)$. 
%(B): the Poincar\'{e} dual interchanging boundary components by vertices, hence producing a unicellular map. 
%}
% \label{F:dual}
%\end{figure}
%%%
%%%%%%%%%%%%%%%%%%%%%%%%%%%%%%%%%%%%%%%%%%%%%%%%%%%%%%%%%%%%%%%%%%%%%%%%%%%%%%%%%%%%%%%%%%%%%
%%%

A bijection is derived by Chapuy in \cite{Chapuy:11} between a unicellular map of 
genus $g$ together with a distinguished trisection and a unicellular map of genus $g-k$,
together with $2k+1$ labeled vertices. This is facilitated by successively slicing a vertex into three new
vertices via the distinguished trisection. Since the trisection can persist through finitely
many slicings, say $k$ times, the process produces $2k+1$ labeled vertices, reducing the genus
of the map by exactly $k$. Gluing is the inverse to slicing: gluing a set of $2k+1$ labeled
vertices in a unicellular map increases the genus of the map by $k$. By construction, both,
slicing and gluing processes preserve unicellularity.

In a unicellular map we have two orders, $<_\gamma$ and $<_\sigma$. The former accounts for
the length of the tour of the boundary component and the latter is induced by counterclockwise
rotation around a vertex, i.e.~from $i$ to $\sigma(i)$. A half-edge $\tau$ is called a {\it trisection} if
$\tau\le \sigma^{-1}(\tau)$ and $\tau$ is not the minimum half-edge of the vertex. A
trisection marks the nontrivial order discrepancies between $\sigma$ and $\gamma$. It is
topologically relevant, since the number of trisections in any unicellular map having genus $g$
equals exactly $2g$ \cite{Chapuy:11}. In particular, any unicellular map having genus
$0$, i.e.~a planar tree, has no trisections.

Let $\bar v$ be a vertex with trisection $\tau$. We slice as follows \cite{Chapuy:11}: 
let $a_1$ denote the minimum half-edge of $\bar v$, let $a_2$ denote the minimum half-edge
between $a_1$ and $a_3=\tau$ (rotating clockwise), that is larger than $a_3$ with respect to
$<_\gamma$. 
%%%
%%%%%%%%%%%%%%%%%%%%%%%%%%%%%%%%%%%%%%%%%%%%%%%%%%%%%%%%%%%%%%%%%%%%%%%%%%%%%%%%%%%%%%%%%%%%%
%%%
%\begin{figure}[t]
%\begin{center}
%  \includegraphics[width=\columnwidth]{glu_sli.eps}
%\end{center}
%\caption{Slicing and gluing: the $\bar\gamma$-traversal (A) from $a_3$ to $a_2$
%  visits the same sectors as the $\gamma$-traversal (B) from $a_1$ to $a_2$.
%  However, the order of $a_2$ and $a_3$ flip: (B) $a_1<_\gamma a_2 <_\gamma a_3$,
%  while in (A)  $a_1<_{\bar \gamma} a_3 <_{\bar \gamma} a_2$.
%} \label{F:glu_sli}
%\end{figure}
%%%
%%%%%%%%%%%%%%%%%%%%%%%%%%%%%%%%%%%%%%%%%%%%%%%%%%%%%%%%%%%%%%%%%%%%%%%%%%%%%%%%%%%%%%%%%%%%%
%%%
%Note that, by construction, $a_1$ and $a_2$ are the minimum half-edges of the new vertices
%$v_1,v_2$, $a_3$, however is either the minimum half-edge of $v_3$ or again a trisection.
%In the former case we call $\tau$ a trisection of {\it type I} or {\it type II}, otherwise.
Iterating this means to further slice $v_3$ until $\tau$ becomes its minimum half-edge
and we derive the mapping $\Xi( \mathfrak{m}_g, \tau )=(\mathfrak{m}_{g-k},v_1, \dots, v_{2k+1})$,
where $k$ denotes the number of slicings. The new vertices $v_1, \dots, v_{2k+1}$ are by construction
ordered by their minimum half-edges with respect to $<_\gamma$. Furthermore, $\Xi$ is bijective.

The process can be reversed, i.e.,  
$\Lambda(\mathfrak{m}_{g-k}, v_1, \dots, v_{2k+1})=(\mathfrak{m}_{g},\tau)$, derived
by successively gluing $v_1, \dots, v_{2k+1}$ in the following fashion: we first glue the last
three vertices $v_{2k-1}$, $v_{2k}$ and $v_{2k+1}$ obtaining a new vertex, $w$, with an (intermediate)
trisection. Continuing this, by gluing the next two vertices with $w$, produces a unicellular
map $\mathfrak{m}_{g}$ together with the trisection $\tau$. 
Iterated slicings and selecting new trisections as needed, generates a unicellular
map having genus $0$, i.e.~a planar tree.
%%%
%%%%%%%%%%%%%%%%%%%%%%%%%%%%%%%%%%%%%%%%%%%%%%%%%%%%%%%%%%%%%%%%%%%%%%%%%%%%%%%%%%%%%%%%%%%%%
%%%
\begin{definition}{\bf (blueprint)}
  Suppose $\mathfrak{m}_g$ is a unicellular map of genus $g$ having $m$ edges together with
  the trisection $\tau$.
A $(\mathfrak{m}_g,\tau)$-blueprint is a sequence 
$$
( (\mathfrak{m}_{g}, \tau), (\mathfrak{m}_{g_1}, \tau_1,), \ldots,
(\mathfrak{m}_{g_{r-1}}, \tau_{r-1}), (\mathfrak{m}_{0},\varnothing)),
$$
where $(\mathfrak{m}_{g_{i+1}},V_{i+1})  = \Xi(\mathfrak{m}_{g_i},\tau_i)$,
for $0\le i <r$ and $g_r=0,\tau_r=\varnothing$.  
\end{definition}
%%%
%%%%%%%%%%%%%%%%%%%%%%%%%%%%%%%%%%%%%%%%%%%%%%%%%%%%%%%%%%%%%%%%%%%%%%%%%%%%%%%%%%%%%%%%%%%%%
%%%
\begin{remark}
  {\rm
    In the following we shall refer to the blueprint of the dual of a pk-structure simply as the
blueprint of the pk-structure. In Fig.~\ref{F:path} we display the blueprints of a matching
having genus $2$.} 
 \end{remark}
  
\begin{remark}
{\rm Note that for two unicellular maps having genus $g$, the number of blueprints is not
necessarily equal. Furthermore, all blueprints can be constructed by considering all
trisections at each step, respectively. The cost of finding all blueprints is exponential
in $g$, however, the genus of RNA structures found in databases is typically less
than $3$.}
\end{remark}
%%%%%%%%%%%
%%%%%%%%%%%%%%%%Show all the paths%%%%%%%%%%%%%%%
%%%%%%%%%%%

%\begin{figure}[ht]
%\begin{center}
%\includegraphics[width=\columnwidth]{path_simple}
%\end{center}
%\caption{ All blueprints of a matching of genus $2$ as 
%  well as all of its associated $\lambda$-matchings induced by all
%  possible blueprints. In the insert, we show how the bottom two
%  $\lambda$-matchings are induced:
%  red circles reference the first
%  and blue circles the second slicing, respectively.  
%}
%\label{F:path}
%\end{figure}

%%%
%%%%%%%%%%%%%%%%%%%%%%%%%%%%%%%%%%%%%%%%%%%%%%%%%%%%%%%%%%%%%%%%%%%%%%%%%%%%%%%%%%%%%%%%
%%%
\section{$\lambda$-structures}\label{S:bi}
%%%
%%%%%%%%%%%%%%%%%%%%%%%%%%%%%%%%%%%%%%%%%%%%%%%%%%%%%%%%%%%%%%%%%%%%%%%%%%%%%%%%%%%%%%%%%%%%%
%%%

We establish now a bijection between unicellular maps $\mathfrak{m}_{g}$ together with
blueprints and a particular class of labeled, planar trees. The labels will
allow to recover the blueprints. Note that, since some vertices may have been involved
in subsequent slicing events, they are not present in the tree.

%%%
%%%%%%%%%%%%%%%%%%%%%%%%%%%%%%%%%%%%%%%%%%%%%%%%%%%%%%%%%%%%%%%%%%%%%%%%%%%%%%%%%%%%%%%%%%%%%
%%%
\begin{definition}{\bf ($\lambda$-tree)}
  A $\lambda$-tree, $\mathfrak{m}_0^{(\sigma_v)_v}$ is a rooted, planar tree with boundary component
  $\gamma$, in which a vertex $v$ carries the label $\sigma_v\in \mathbb{F}_2^r$ such that\\
  $\bullet$ $\sum_{v,h} \sigma_v|_h = 2g+r$, where $1\le r\le g$,\\
  $\bullet$ $|\{\sigma_v \mid \sigma_v|_h=1 \}|\equiv 1 \mod 2$,\\
  $\bullet$ $\vert \{ v_{\sigma_v} \mid \sigma_v|_h=\sigma_v|_i=1, \; h<i\}\vert \le 1$
  and any vertex, whose label $\sigma_v$ satisfies $\sigma_v|_h=\sigma_v|_i=1$
  for $h<i$, is minimal in $\gamma$ within the set of all the vertices, having label $\sigma_v$ where
  $\sigma_v|_i=1$.
\end{definition}
%%%
%%%%%%%%%%%%%%%%%%%%%%%%%%%%%%%%%%%%%%%%%%%%%%%%%%%%%%%%%%%%%%%%%%%%%%%%%%%%%%%%%%%%%%%%%%%%%
%%%
We shall call the minimum vertex $v$ with the property $\sigma_v|_i=1$ transitional of level $i$
and regular, otherwise. Note that a vertex can be transitional for multiple indices. 
%%%
%%%%%%%%%%%%%%%%%%%%%%%%%%%%%%%%%%%%%%%%%%%%%%%%%%%%%%%%%%%%%%%%%%%%%%%%%%%%%%%%%%%%%%%%%%%%%
%%%
\begin{lemma}\label{L:welldefined}
  Any unicellular map, $\mathfrak{m}_g$, together with a blueprint, $\mathfrak{p}$, induces a
  unique $\lambda$-tree, $\mathfrak{m}_0^{(\sigma_v)_v}$.
\end{lemma}
%%%
%%%%%%%%%%%%%%%%%%%%%%%%%%%%%%%%%%%%%%%%%%%%%%%%%%%%%%%%%%%%%%%%%%%%%%%%%%%%%%%%%%%%%%%%%%%%%
%%%
\begin{proof}
Suppose we are given a unicellular map $\mathfrak{m}_g$ together with a slice-path consisting of
$r$ slicings. $\mathfrak{p}$ generates a planar tree, $\mathfrak{m}_0$, whose vertices we color
red if they were involved in some slicing and black, otherwise. $\mathfrak{p}$ furthermore specifies
the sequence of slicings and each slicing reduces topological genus by at least one, whence $1\le r\le
g$.

We label the black vertices by $(0,\dots,0)\in \mathbb{F}_2^r$ and each red vertex obtains a $1$ in
its $s$th-coordinate if and only if it was involved in the $s$th-slicing. This produces a set of
intermediate labels $(\sigma_v')_v$.

We proceed by reducing the $(\sigma_v')_v$: starting from $i=r$ to $i=2$, we set
 \begin{equation*}
  \sigma_v=
  \begin{cases}
    (0,\dots,x_i=1,x_{i+1},\dots,x_r) \quad  \text{\rm if $v$ is regular at $i$}\\
    \sigma_v \quad \text{\rm otherwise.}
  \end{cases}
 \end{equation*}
  Accordingly, any $\gamma$-minimal vertex of a fixed slicing retains a $1$ iff it came
  from a previously sliced vertex and, as a result, irrespective of how many times the
  vertices were sliced, the total number of $1$ entries in all coordiantes of all lables
  equals:
  \begin{equation*}
    \sum_{\sigma_v} \vert \{h\mid \sigma_v|_h=1\}\vert = \sum_{i=1}^r (2\Delta g_i+1) = 2g+r,
  \end{equation*}
    where $\Delta g_i$ is the decrease of genus in the $i$th slicing. As a result, we have
    $\sum_{v,h} \sigma_v\vert_h = 2g+r$.

    By construction, the total number of $1$ entries found as $s$th-coordinates equals
    the total number of vertices obtained via the $s$th slicing, whence
    $|\{\sigma_v \mid \sigma_v|_s=1 \}|\equiv 1 \mod 2$.

  Finally, consider two fixed indices $h<i$. Only the label of the $\gamma$-minimal vertex
  amongst the vertices of the $i$th slicing can have an entry $1$ as $h$th coordinate, whence
  $\vert \{ v \mid \sigma_v|_h=\sigma_v|_i=1, \; h<i\}\vert \le 1$ and the proof of the lemma
  is complete. In Fig.~\ref{F:remove_label} we illustrate the idea of the proof.
\end{proof}
%In Fig.~\ref{F:remove_label} we illustrate the idea of the proof.

%%%%%%%%%%%
%%%%%%%%%%%%%%%%example of labeled sec%%%%%%%%%%%%%%%
%%%%%%%%%%%
%\begin{figure}[t]
%\begin{center}
%\includegraphics[width=\columnwidth]{label}
%\end{center}
%\caption{\small The idea of the proof of Lemma~\ref{L:welldefined}: the labels record the
%  successive slicings and are finally reduced, such that only transitional vertices carry
%  the information about previous slicing events.
%}
%\label{F:remove_label}
%\end{figure}
%%%%%%%%%%%
%%%%%%%%%%%%%%%%%%%%%%%%%%%%%%%%%%%%
%%%%%%%%%%%
Having established the mapping from unicellular maps of genus $g$ together with
blueprints into $\lambda$-trees, we proceed by showing that this mapping is bijective. 
%%%
%%%%%%%%%%%%%%%%%%%%%%%%%%%%%%%%%%%%%%%%%%%%%%%%%%%%%%%%%%%%%%%%%%%%%%%%%%%%%%%%%%
%%%
\begin{proposition}\label{P:bi}
Suppose $\mathfrak{m}_g$ is a unicellular map with genus $g$ and blueprint $\mathfrak{p}$, 
then we have the bijection
\begin{equation*}
\xi \colon (\mathfrak{m}_g, \mathfrak{p}) \mapsto \mathfrak{m}^{(\sigma_v)_v}_0.
\end{equation*}
\end{proposition}
%%%
%%%%%%%%%%%%%%%%%%%%%%%%%%%%%%%%%%%%%%%%%%%%%%%%%%%%%%%%%%%%%%%%%%%%%%%%%%%%%%%%%%
%%%
\begin{proof}
Following the slice-path, $\mathfrak{p}$, the fatgraph $\mathfrak{m}_g$ is sliced to
the tree $\mathfrak{m}_0$ with labels $(\sigma_v)_v$, $\mathfrak{m}^{(\sigma_v)_v}_0$.
By Lemma~\ref{L:welldefined}, $\mathfrak{m}^{(\sigma_v)_v}_0$ is a $*$-tree, whence
$\xi(\mathfrak{m}_g)=\mathfrak{m}^{(\sigma_v)_v}_0$ is well-defined.

To show that $\xi$ is bijective, let $\mathfrak{m}_0^{(\sigma_v)_v}$ be an arbitrary
$*$-tree such that $\sigma_v\in\mathbb{F}_2^r$. We shall construct $\xi^{-1}$ as
follows:

Let $V_r = \{v  \mid \sigma_v|_r = 1\}$ denote the set of vertices whose labels satisfy
$\sigma_v|_r=1$. Lemma \ref{L:welldefined} guarantees that the number of vertices contained
in $V_r$ is odd. Then we glue 
$$
\Lambda\colon (\mathfrak{m}_0, V_r ) \to (\mathfrak{m}_{g_{r-1}}, \tau_{r-1} ), 
$$
creating the trisection, $\tau_{r-1}$, where $g_{r-1} = (|V_r|-1)/2$.

By construction $\mathfrak{m}_{g_{r-1}}$ is a unicellular map with two types of vertices:
those which correspond to $\mathfrak{m}_0$-vertices (type $1$) and a unique new vertex
obtained by gluing the set $V_r$ (type $2$).
We now label the $\mathfrak{m}_{g_{r-1}}$-vertices with labels $\sigma_v^{r-1}$
using the labels $\sigma_v$ of $\mathfrak{m}_0$ as follows:
For any type $1$ vertex we set $\sigma^{r-1}_v|_i=\sigma_v|_i$ for $1\le i\le r-1$
and $\sigma^{r-1}_v|_r=0$. The unique type $2$ vertex obtains the label
$$
\sigma_v^{r-1}|_i=\sum_{v\in V_r}\sigma_v|_i, \quad \forall 1\le i \le r-1, \quad \quad 
\sigma_v^{r-1}|_r = 0. 
$$
In Fig.~\ref{F:blueprint} we illustrate the idea of the proof.
\end{proof}
%%%
%%%%%%%%%%%%%%%%%%%%%%%%%%%%%%%%%%%%%%%%%%%%%%%%%%%%%%%%%%%%%%%%%%%%%%%%%%%%%%%%%%%%%%%%%%%%%
%%%
%\begin{figure}[t]
%\begin{center}
%  \includegraphics[width=\columnwidth]{blueprint}
%\end{center}
%\caption{Reconstructing the blueprint from a $\lambda$-tree: we consider the set of all vertices
%  such that $\sigma_v|_2=1$ (blue circles) and glue and relabel the
%  vertices as in the proof of Proposition\ref{P:bi} in the Supplemental Materials.
% This generates the unicellular map $\mathfrak{m}_1$ with the glued vertex
%  labeled by $(1,0)$, carrying the distinguished trisection $8$.
%  Iteration of this process, produces $\mathfrak{m}_2$ together with a distinguished
%  trisection. We have thus constructed from the $\lambda$-tree the blueprint
%  $((\mathfrak{m}_{2}, \tau), (\mathfrak{m}_{1}, \tau_1,), (\mathfrak{m}_0,\varnothing))$.
%} \label{F:}
%\end{figure}
%%%
%%%%%%%%%%%%%%%%%%%%%%%%%%%%%%%%%%%%%%%%%%%%%%%%%%%%%%%%%%%%%%%%%%%%%%%%%%%%%%%%%%%%%%%%%%%%%
%%%
A $\lambda$-tree corresponds to a unicellular map of genus $g$, together with a blueprint. As a planar
tree it corresponds by Poincar\'{e} duality to a noncrossing matching. 
We call the edge-labeled Poincar\'{e} dual of a $\lambda$-tree a $\lambda$-matching, $M_0^{(\sigma_a)_a}$,
whose edge-labels $(\sigma_a)_a$ are induced by the vertex labels of the $\lambda$-tree $(\sigma_v)$ as
follows: there are $n+1$ vertices and $n$ edges in a planar tree, whence a vertex corresponds
uniquely to an edge (except of the root). To this end, we consider the edge which connects $v$ to
its parent, see Fig.~\ref{F:vtoe}. Note that, by construction, the root of the planar tree is never
involved in slicings and accordingly labeled $0\in\mathbb{F}_2^r$.
%%%%%%%%%%%
%%%%%%%%%%%%%%%%example of labeled sec%%%%%%%%%%%%%%%
%%%%%%%%%%%

%\begin{figure}[ht]
%\begin{center}
%  \includegraphics[width=0.9\columnwidth]{vtoe}
%\end{center}
%\caption{From a $\lambda$-tree with labeled vertices to a $\lambda$-matching with labeled arcs.
%  (A) a $\lambda$-tree with labeled vertices. (B) shifting the labels from vertices 
%  to edges. (C) a $\lambda$-matching with labeled arcs obtained from (B) by taking the
%  Poincar\'{e} dual.
%}
%\label{F:vtoe}
%\end{figure}

%%%%%%%%%%%
%%%%%%%%%%%%%%%%%%%%%%%%%%%%%%%%%%%%
%%%%%%%%%%%

A $\lambda$-matching induces an arc-labeled secondary structure by insertion of unpaired vertices.
To facilitate this we employ the index map $I\colon \mathbb{N}_0 \to \mathbb{N}_0$, $i\to j$.
We interpret $I$ as to specify how many unpaired vertices follow the $M_0^{(\sigma_a)_a}$-vertex,
$i$. Starting with the left-endpoint of the rainbow and we have $I^{-1}(j+1) - I^{-1}(j) -1$
such vertices.
%%%
%%%%%%%%%%%%%%%%%%%%%%%%%%%%%%%%%%%%%%%%%%%%%%%%%%%%%%%%%%%%%%%%%%%%%%%%%%%%%%%%%%%%%%%%%%%%%%
%%%

\begin{definition}
  A $\lambda$-structure $S^{(\sigma_a)_a}$ is a noncrossing diagram with
  labeled arcs, $\sigma_a$, that when removing all unpaired
  vertices, induces a $\lambda$-matching.  
\end{definition}

%%%
%%%%%%%%%%%%%%%%%%%%%%%%%%%%%%%%%%%%%%%%%%%%%%%%%%%%%%%%%%%%%%%%%%%%%%%%%%%%%%%%%%%%%%%%%%%%%%
%%%

Note that slicing induces transpositions of the boundary component. Thus, taking the Poincar\' e
dual, these transpositions permute the backbone of $M_g^{(\sigma_a)_a}$. Let $\rho=\rho(\mathfrak{p})$
denote resulting permutation of the backbone, where $\rho(0)=0$, then $I^{\rho}(h)=I(\rho^{-1}(h))$
is the induced index map for the pk-structure corresponding to the $\lambda$-structure and we have the
following situation, see Fig.~\ref{F:comm1}. 
{\small $$
\diagram
(S_{g,n},\mathfrak{p}) \dline \rline & (M_g,\mathfrak{p},I^\rho) \dline \rline^{\text{dual}} &
(\mathfrak{m}_g,\mathfrak{p},I^\rho) \dline \\
S^{(\sigma_a)_a} & \lline  (M_0^{(\sigma_a)_a},I) & \lline^{\text{dual}}  
(\mathfrak{m}_0^{(\sigma_v)_v},I)
\enddiagram
$$}

%\begin{figure}[ht]
%\begin{center}
 % \includegraphics[width=0.9\columnwidth]{comm1}
%\end{center}
%\caption{A pk-structure together with a blueprint maps to a unique $\lambda$-structures. }
%\label{F:comm1}
%\end{figure}

%%%
%%%%%%%%%%%%%%%%%%%%%%%%%%%%%%%%%%%%%%%%%%%%%%%%%%%%%%%%%%%%%%%%%%%%%%%%%%%%%%%%%%%%%%%%%%%%%%%%%%%
%%%
\section{The grammar}\label{S:grammar}
%%%
%%%%%%%%%%%%%%%%%%%%%%%%%%%%%%%%%%%%%%%%%%%%%%%%%%%%%%%%%%%%%%%%%%%%%%%%%%%%%%%%%%%%%%%%%%%%%%%%%%%
%%%

We can now derive the CFG for $\lambda$-structures, {\it RNAFeatures$^*$}. The key
issue here will be the compatibility of the arc labels with the production rules. 

%What is CFG, simple version of CFG for classic secondary structure
Let us briefly review the CFG for secondary structures, implied by the recursion of \cite{Waterman:78a},
i.e.~the tuple $G=(V, \Sigma, R,
I)$, where $V,\Sigma,S$ are sets of nonterminal symbols, terminal symbols and production rules,
respectively. A production rule is a mapping from
$V$ to $(\Sigma\bigcup V)^*$, where the asterisk represents the Kleene star operation
and $S$ is an initial symbol. 
A derivation starts with the initial symbol $I$,  and in each step, replaces one of 
the nonterminal symbols by one of the productions. 
For secondary structures, the grammar consists of 
$V= \{S, L\}$, $\Sigma = \{d,d',s\}$, and the two production rules are
\begin{equation} \label{E:sec_grammar}
S \to \epsilon \quad \text{or} \quad  LS \quad \text{or} \quad sS, \quad \quad  L \to dSd'. 
\end{equation}
Here $S \to \epsilon$ is the $\epsilon$-production. The terminal symbols $d$, $d'$ represent pairs 
of vertices forming arcs and $s$ an unpaired vertex, respectively. The nonterminal symbol $S$
denotes an arbitrary structure and $L$ is an irreducible structure, i.e.~covered by an external
arc. Via the initial symbol $S$ this grammar is unambiguous, i.e.~each structure has a unique
derivation path.

%Introducing new terminals and nonterminals, labels 
We shall now extend the grammar to $\lambda$-structures. To this end we introduce
new nonterminals, $S_{(t_i)_i}^{(\ell_i)_i}$ and $L_{(t_i)_i}^{(\ell_i)_i}$. The nonterminal 
symbols stand for sets of substructures $N=\{S, L\}$, with labeled arcs $(\sigma_a)_a$
such that $\ell_h = \sum_{\sigma_a\in N} \sigma_a|_h$,
$1\le h \le r$. Aside from $s$ we have new terminals $\{d^{\sigma_a}, d'^{\sigma_a}\}$
where $d^{\sigma_a}, d'^{\sigma_a}$ denote $\sigma_a$-labelled arcs. We set $t_i=1$ if
and only if the corresponding nonterminal contains a transitional arc for $i$ and
$t_i=0$, otherwise.

%discuss enhanced rules in detail
By construction, the $\ell_h$ will evolve depending on the presence of transitional arcs,
that is $(\ell_h)_h$ may change in more than one coordinate if a transitional arc is
present, or in at most one coordinate, otherwise. $(t_i)_1^r$ can be viewed as an
indicator determining which production rule applies.

{\it Decomposition:} this means to decompose a nonterminal symbol $S^{(\ell_i)_i}_{( t_i)_i}$
into one left- and one right-symbol, together with consistent labelings.
A labeling for $S^{(\ell_i)_i}_{( t_i)_i} \to sS^{(\rho_i)_i}_{( p_i)_i}$ is consistent if
$\ell_i=\rho_i$ and $t_i=p_i$. Furthermore, consistency for 
$S^{(\ell_i)_i}_{( t_i)_i}\to L^{(\rho_i)_i}_{( p_i)_i} S^{(\xi_i)_i}_{( q_i )_i}$ means \\
(a) $t_i=1$ and $\rho_i>0$ implies $p_i=1$, \\
(b) $t_i=1$ and $\rho_i=0$ implies $\xi_i>0$ and $q_i=1$\\
(c) $t_i=0$                implies $p_i=q_i=t_i=0$.\\
This is a consequence of transitional arcs being induced by transitional vertices by means
of Poincar\' e duality and the latter being minimal in the boundary component. As a result
a transitional arc needs to be placed in the left-nonterminal as long as the latter contains
any arcs $a$ with the property $\sigma_a|_i=1$.

As a result we have the production rule:
\begin{equation}  \label{E:rule_can}
S^{(\ell_i)_i}_{( t_i)_i} \to \epsilon \quad \text{or} \quad L^{(\rho_i)_i}_{( p_i)_i} 
S^{(\xi_i)_i}_{( q_i )_i} 
\quad \text{or} \quad  sS^{(\ell_i)_i}_{( t_i)_i}, 
\end{equation}
satisfying $\rho_i+\xi_i = \ell_i$, for $\rho_i$, $\xi_i\ge 0$, $1\le i \le r$. 
In case of $t_i=1$, if $\rho_i>0$, we have $p_i=1$ and $q_i=0$. Otherwise, if $\rho_i=0$,
the $i$th trans-arc is distributed to the right and we have $p_i=0$ and $q_i=1$. Finally,
if $t_i=0$, both, $p_i=q_i=0$, see Fig.~\ref{F:grammar1}.

%\begin{figure}[ht]
%\begin{center}
%\includegraphics[width=\columnwidth]{grammar1}
%\end{center}
%\caption{\small Decomposition: any transitional arc is placed in the left-nonterminal as
%  long as the latter contains any arcs $a$ with the property $\sigma_a|_i=1$.  
%}
%\label{F:grammar1}
%\end{figure}

{\it Arc-removal:} this means to act on the nonterminal $L_{(t_i)_i}^{(\ell_i)_i}$ by
removing its outer arc, $a$. The key point here is to provide this arc with a label,
$\sigma_a$ consistent with the label associated with the generated nonterminals.
Given $L_{(t_i)_i}^{(\ell_i)_i}$ any arc-removal based production rule
generates the nonterminal $S_{(t'_i)_i}^{(\ell'_i)_i}$, where
$t_i'\le t_i$ and $\ell_i-\ell'_i\le 1$. This implies for the label of the removed arc
$\sigma_a|_h=(\ell_h-\ell'_h)$, $1\le h \le r$. In case $a$ is a trans-arc for $s_1,\dots s_k$, 
we have $0=t'_{s_j}<t_{s_j}=1$, or $0=t'_j=t_j$, otherwise. 
Accordingly, we have:
\begin{equation} \label{E:rule_rem}
L^{(\ell_i)_i}_{( t_i)_i} \to d^{\sigma_a}S^{(\ell_i)_i-\sigma_a}_{( p_i)_i}d'^{\sigma_a}. 
\end{equation}

%New grammar, proof it generating all $*$ sec. May move to supplement
The terminal symbols $\Sigma=\{ d^{\sigma_j}, d'^{\sigma_j}, s \}$, nonterminal symbols 
$V \ = \{ S_{(t_i)_i}^{(\ell_i)_i}, L_{(t_i)_i}^{(\ell_i)_i} \}$, production rules $R$ given by 
eq.~[\ref{E:rule_can}], eq.~[\ref{E:rule_rem}], together with the initial condition
$I=S^{(\ell_i)_i}_{(t_i)_i}$ with $\sum_i^r \ell_i = 2g+r$ and $t_i=1$,
for $1\le i \le r$, define the context-free grammar $G=(V,\Sigma, R, I)$. 

%%%%
%%%%%%%%%%%%%%%%%%%%%%%%%%%%%%%%%%%%%%%%%%%%%%%%%%%%%%%%%%%%%%%%%%%%%%%%%%%%%%%%%%%%%%%%%
%%%%

\begin{proposition}
  The context-free grammar $G$ discussed above generates all $\lambda$-structures uniquely, or
  equivalently, all pk-structures together with a blueprint.
\end{proposition}

{\it Proof:}
Consider a nonterminal symbol $S^{(\ell_i)_i}_{(t_i)_i}$ where the $\ell_i$ are positive 
odd integers, with $\sum_i^r \ell_i = 2g+r$ and $t_i=1$. Any nonterminal symbol  
is uniquely generated by $G$ since the product of each production rule is unique.
Furthermore, no labeled arc with $\sigma_a|_h=1$ can be further added to the structure
when $t_h$ is set to be $1$. This guarantees that the trans-arc of level $h$ is  
minimal among those arcs with $\sigma_a|_h=1$. Thus the structures generated from  
$S^{(\ell_i)_i}_{(t_i)_i}$ are $\lambda$-structures.
As for its generative power, an arbitrary a $\lambda$-structure, is translated via $G$ into 
a unique sequence of terminal symbols, whence the proposition.

%%%%
%%%%%%%%%%%%%%%%%%%%%%%%%%%%%%%%%%%%%%%%%%%%%%%%%%%%%%%%%%%%%%%%%%%%%%%%%%%%%%%%%%%%%%%%%
%%%%

%In practice we use a grammar with loop feature, see \cite{X} and supplement
The grammar $G$ for $\lambda$-structure is designed based on the simple grammar 
for secondary structure by assigning labels to symbols. However, the latter evaluates
arcs against unpaired vertices, lacking the capability of differentiating base pairs
and unpaired bases in the context of specific loop types.

We accordingly present a grammar for $\lambda$-structures, {\it RNAFeatures$^*$} used in practice that
differentiates loop types, that is an extension of the grammar {\it RNAFeatures} \cite{Bio:method}.
The grammar allows us to differentiate 
loop-types. First we have the terminal symbols $\Sigma = \{A, C, G, U, R\}$ representing the respective
RNA nucleotides and the endpoint of the rainbow. The nonterminal symbols are 
\begin{equation*}
\begin{split}
V = \{ \lambda_g\_structure,  
ExteriorLoop^{(\ell_i)_i}_{( t_i)_i}, Stack^{(\ell_i)_i}_{( t_i)_i}, Weak^{(\ell_i)_i}_{( t_i)_i}, \\
 HairpinLoop^{\sigma_a}_{(0)_i}, InteriorLoop^{(\ell_i)_i}_{( t_i)_i}, BulgeLeft^{(\ell_i)_i}_{( t_i)_i}, \\
BulgeRight^{(\ell_i)_i}_{( t_i)_i}, MultiLoop^{(\ell_i)_i}_{( t_i)_i},  MLComponent^{(\ell_i)_i}_{( t_i)_i}, \\
MLComponents^{(\ell_i)_i}_{( t_i)_i}, SingleStrand \}. 
\end{split}
\end{equation*}

The nonterminal symbols are referenced via both: their semantics and their arc-label information
of the class of structures it represents.
$\lambda_g$-structure are $\lambda$-structures corresponding to a pk-structure of genus $g$, together 
with a blueprint.
For nonterminal symbols the $N^{(\ell_i)}_{((t_i)_i)}$, where 
$$
N=\{ExteriorLoop, Stack, \ldots, MLComponents\}, 
$$ 
the index $(\ell_i)$ represents the sum of the arc-labels contained in the class,
such that $t_i=1$ if and only if the transitional arc of level $i$ is present.
We distinguish the following classes: 
\begin{itemize} 
\item $ExteriorLoop$ denotes the class that is immediately nested in the rainbow, 
\item $Stack$ denotes the class covered by parallel arcs, while $Weak$ represents the class covered by a isolated arc, 
\item $HairpinLoop$ represents the class containing no arcs, immediately nested in $Weak$,
      while $InteriorLoop$, $BulgeLeft$, $BulgeRight$ have one and $MultiLoop$ have more than one such arc. 
     The classes $InteriorLoop$, $BulgeLeft$ and $BulgeRight$ differ in the number of unpaired 
     bases in the two intervals formed between the outer arc and the immediately nested arc.
     $BulgeLeft$ has the interval to the right empty and the intervals to the left non-empty,
     $BulgeRight$ is defined analogously and in $InteriorLoop$ both these intervals are non-empty, 
\item $MLComponents$ and $MLComponent$ are classes contained in a multi-loop: $MLComponent$ has 
exactly one braching and $MLComponents$ has at least two, 
\item $SingleStrand$ denotes the class consisting of a sequence of unpaired bases. 
\end{itemize}

The nonterminal symbols are listed in Fig.~\ref{F:nonterminals}

\begin{itemize}
\item $ini_r$: here we have $\sum_{0<i\le r} \ell_i+\sigma_a = 2g+r$. In case of $\sigma_a \neq (0)_i$,
the rainbow is always a transitional arc, hence $t_i = 0$ if $\sigma_a|_i \neq 0$,  
\item $exl_1$ emits the signal $\epsilon$ of the $\lambda$-structure,
\item $exl_2(x, (\ell_i)_i, (t_i)_i)$ cuts off the leftmost unpaired base and keeps all the labels as well as transitional 
arcs in the right nonterminal, where $x\in \{A,C,G,U\}$.  $co_1(x, (\ell_i)_i, ( t_i)_i)$, 
$ss_1(x)$ and $ss_2(x)$ are defined analogously,
\item $exl_3((\ell_i)_i, (\rho_i)_i, ( t_i)_i)$ decompose into two classes, distributing any labeled arcs
and in particular any trans-arcs, whence $(\ell_i)_i =  (\rho_i)_i + (\xi_i)_i$. 
In case of $t_i=1$, if $\rho_i>0$, we have $p_i=1$ and $q_i=0$. Otherwise, if $\rho_i=0$,
the $i$th trans-arc is distributed to the RHS and we have $p_i=0$ and $q_i=1$. Finally,
if $t_i=0$, both, $p_i=q_i=0$.  The rules $cs_1((\ell_i)_i, (\rho_i)_i, (t_i)_i)$,
$cs_2((\ell_i)_i, (\rho_i)_i, (t_i)_i)$ and $cs_3((\ell_i)_i, (\rho_i)_i, (t_i)_i)$ are defined analogously, 
\item $st_1(x,y,(\ell_i)_i, (t_i)_i, \sigma_1)$ remove a base pair $(x,y)$ with the arc label $\sigma_a$,
$x,y\in \{A,C,G,U\}$, inducing the labels $(\ell_i)_i-\sigma_a$ for the class produced.
If $\sigma_a|_i\neq 0$ and $t_i=1$, the arc $(x^{\sigma_a}, y^{\sigma_a})$ is a trans-arc of level 
$i$ and we set $p_i=0$ accordingly. The rules $st_2(x,y,(\ell_i)_i, (t_i)_i, \sigma_a)$, $hl(x,y, \sigma_a)$,
$il(x,y,(\ell_i)_i, (t_i)_i, \sigma_a)$, $bl(x,y,(\ell_i)_i, (t_i)_i, \sigma_a)$, $br(x,y,(\ell_i)_i, (t_i)_i,
\sigma_a)$ as well as $ml(x,y,(\ell_i)_i, (t_i)_i, \sigma_a)$, $x,y\in \Sigma$ are defined analogously.  
\end{itemize}

We display the production rules in Fig.~\ref{F:rules}. 

%score function, compute score with $(\theta, S)$
Let $S$ be a structure and $\theta$ be an RNA sequence. A CFG allows to assign a score to a
pair $(S,\theta)$, derived from scores assigned to the production rules encountered. Namely,
for each $(S,\theta)$ the CFG provides a unique parse tree and the score of $(S, \theta)$ is
then computed by means of the scores of the respective production rules. Scores
based on a thermodynamic energy model produces the minimum free energy (mfe) model 
\cite{Mathews:99, Mathews:04, Turner:10}
and a probabilistic scoring scheme derived from data bases of RNA sequences and structures
leads to a SCFG $G=(V,\Sigma, R, S, P)$. In the latter case, sequence--structure pairs of the
data set allow to assign scores to the production rules, which in turn induce a score of an
arbitrary pair $(S,\theta)$.

%%%%
%%%%%%%%%%%%%%%%%%%%%%%%%%%%%%%%%%%%%%%%%%%%%%%%%%%%%%%%%%%%%%%%%%%%%%%%%%%%%%%%%%%%%%%%%
%%%%
\section{The SCFG of RNA pseudoknot structures}
%%%%
%%%%%%%%%%%%%%%%%%%%%%%%%%%%%%%%%%%%%%%%%%%%%%%%%%%%%%%%%%%%%%%%%%%%%%%%%%%%%%%%%%%%%%%%%
%%%%

%What is SCFG
Stochastic context-free grammars (SCFG) were developed for RNA secondary structures
\cite{Eddy:94, Sakakibara:94}. They have been employed for the {\it ab initio} structure
prediction as well as the analysis of observables such as base pairing probabilities or
loop patterns \cite{Knudsen:03, Eddy:00, Eddy:01a}. A SCFG is derived from a CFG by associating
to each production rule $r$ a probability $p_r$. That is, for any nonterminal $A$, suppose that
$r_1, \ldots r_k$ are all derivations of $A$, then we have $\sum_i p_i = 1$.

Let $S_g$ denotes an RNA structure of genus $g$ and $\theta$ be a sequence. For $g=0$, $S_{0}$
is a secondary structure and can be parsed by {\it RNAFeatures}. We construct the parse tree of
$(S_{0},\theta)$ and record the frequency of the production rules.
In case of $g>0$, for any $S_g$-blueprint, $\mathfrak{p}(S_g)$, we consider
the correspondence
\begin{equation} \label{E:bi}
(S_g,\mathfrak{p}(S_g),\theta) \to (S^{(\sigma_a)_a}_{\mathfrak{p}(S_g)}, 
\theta_{\mathfrak{p}(S_g)}), 
\end{equation}
where $S^{(\sigma_a)_a}_{\mathfrak{p}(S_g)}$ is a $\lambda$-structure and
$\theta_{\mathfrak{p}(S_g)}$ is the modified sequence. We construct the parse tree
of {\it RNAFeatures}$^*$
of $(S^{(\sigma_a)_a}_{\mathfrak{p}(S_g)}, \theta_{\mathfrak{p}(S_g)})$ and record the
frequency of the encountered production rules.

If there are $m$ $S_g$-blueprints, the contribution of each pair
$(S^{(\sigma_a)_a}_{\mathfrak{p}(S_g)}, \theta_{\mathfrak{p}(S_g)})$ is normalized by $1/m$
and the frequency of a production rule $R$ equals $f(R)/m$ where $f(R)$
is the frequency of $R$ being applied. After processing all structures and sequences,
the probability of a production rule is computed as the quotient of its frequency
over the sum of all frequencies of its respective derivation. 

%How to compute p(\theta, S_g)
We next consider $\mathbb{P}(S_g, \theta | M)$, the normalized score of an arbitrary
structure, $S_g$, over an arbitrary sequence, $\theta$, as a function of the data set
$M$. We set $\mathbb{P}(S_g,\theta | M) = \sum_{\mathfrak{p}(S_g)} \mathbb{P}
(S^{(\sigma_a)_a}_{\mathfrak{p}(S_g)}, \theta_{\mathfrak{p}(S_g)} | M)$, 
where $\mathbb{P}(S^{(\sigma_a)_a}_{\mathfrak{p}(S_g)},\theta_{\mathfrak{p}(S_g)} | M)$
is calculated by the multiplying the probabilities of the production rules in the
parse tree of $(S^{(\sigma_a)_a}_{\mathfrak{p}(S_g)}, \theta_{\mathfrak{p}(S_g)})$ and the sum
is taken over the set of all $S_g$-blueprints.

%argue P(S_g) = \sum_p P(S*)
The $\mathbb{P}(S_g, \theta | M)$ induce 
$\mathbb{P}(S_g | M) = \sum_{\theta\in \mathcal{Q}^n_4} \mathbb{P}(S_g, \theta |M)$.
Using the fact that for fixed blueprint, $\mathfrak{p}$, we have a bijection
$\beta_\mathfrak{p}\colon Q_4^n\to Q_4^n$, we derive
\begin{eqnarray*}
\mathbb{P}(S_g | M) &=& \sum_{\theta_{\mathfrak{p}(S_g)} \in \mathcal{Q}^n_4}  \sum_{\mathfrak{p}(S_g)} 
\mathbb{P}(S_{\mathfrak{p}(S_g)}^{(\sigma_a)_a},
\theta_{\mathfrak{p}(S_g)} |M)  \\
& = & \sum_{\mathfrak{p}(S_g)} \mathbb{P}(S_{\mathfrak{p}(S_g)}^{(\sigma_a)_a}|M). 
\end{eqnarray*}

%Bolzmann sample a *-sec S^*, permute it back to S_g argue it is Boltzmann sample S_g
Furthermore the SCFG facilitates the sampling of structures $S_g \in \mathbb{S}_{g, n}$
with Boltzmann probability, i.e., $\mathbb{P}(S_g) = \mathbb{P}(S_g | M) /
\sum_{S \in \mathbb{S}_{g, n}} \mathbb{P}(S|M)$.
We employ {\it RNAFeatures$^*$} to Boltzmann sample a $\lambda$-structure $S^*$, and
reconstruct $S_g$ as well as the associated blueprint, $\mathfrak{p}(S_g)$: fixing $S_g$,
suppose there are $m$ blueprints $\mathfrak{p}_i(S_g)$, these correspond to the $m$ distinct
$\lambda$-structures $S^*_i$. Let $\mathbb{S}_n^*$ be the collection of all such $S^*_i$.
In view of  $S_i^* \to (S_g, \mathfrak{p}_i(S_g))$ being a bijection, we have 
$$
\sum_{S_g \in \mathbb{S}_{g,n}} \mathbb{P}(S_g |M) 
=  \sum_{S_g \in \mathbb{S}_{g,n}} \sum_{\mathfrak{p}(S_g)} \mathbb{P}(S_{\mathfrak{p}(S_g)}^{(\sigma_a)_a}|M)
= \sum_{S^* \in \mathbb{S}^*_n} \mathbb{P}(S^* | M)   
$$
Therefore, in view of $\mathbb{P}(S_g) = \sum_{\mathfrak{p}_i(S_g)} \mathbb{P}(S_i^*)$ and
$$
\sum_{\mathfrak{p}_i(S_g)} \mathbb{P}(S_i^*) = 
\frac{\sum_{\mathfrak{p}_i(S_g)} \mathbb{P}
(S^*_i | M)} {\sum_{S^* \in \mathbb{S}^*_n} \mathbb{P}(S^* |M)} = 
  \frac{\mathbb{P}(S_g | M)}
{\sum_{S_g \in \mathbb{S}_{g,n}} \mathbb{P}(S_g | M)}
$$
we have $\mathbb{P}(S_g) = \sum_{\mathfrak{p}_i(S_g)} \mathbb{P}
(S^*_i | M)/ \sum_{S^* \in \mathbb{S}^*_n} \mathbb{P}(S^* |M)$.

As for runtime complexity, the Boltzmann sampler of secondary structures using such a SCFG
has runtime of $O(n^2)$, if structures are ranked according to the sequential order. This can
be improved to $O(n\log n)$ by ranking the structures with the Boustrophedon order
\cite{Ponty:08,Nebel:11}. For fixed topological genus, the label enhancement of {\it RNAFeatures}$^*$
by construction increases only the number of derivations. Thus the Boltzmann sampler for
$\lambda$-structures has runtime complexity $O(n\log n)$.
%%%%
%%%%%%%%%%%%%%%%%%%%%%%%%%%%%%%%%%%%%%%%%%%%%%%%%%%%%%%%%%%%%%%%%%%%%%%%%%%%%%%%%%%%%%%%%
%%%%

\section{Discussion}

%%%%
%%%%%%%%%%%%%%%%%%%%%%%%%%%%%%%%%%%%%%%%%%%%%%%%%%%%%%%%%%%%%%%%%%%%%%%%%%%%%%%%%%%%%%%%%
%%%%

RNA secondary structures have been widely studied using the classic CFG implied by the recursions in
\cite{Waterman:78aa, Waterman:78a}. 
Using RNA sequence and structure data sets, this CFG induces a SFCG. This SCFG has been applied for the
{\it ab initio} prediction of RNA secondary structures \cite{Knudsen:99, Knudsen:03}, calculating the
likelihood of a sequence generated by the model SCFG \cite{Eddy:00}, and detecting RNA genes using
comparative sequence analysis \cite{Eddy:01a}. 
Furthermore, the generation of random secondary structures according to the native distributions of 
certain families of RNA structures is studied in \cite{Nebel:11} in order to identify key structural
motifs in biological RNA secondary structures. This paper allows to extend these studies to
pk-structures. To the best of our knowledge {\it RNAFeatures}$^*$ is the first unambiguous
context-free grammar for pk-structures. Related work has been done by Freiermuth and Nebel on a
CFG of certain, colored trees. These are trees that can be mapped to unicellular maps, however, not
bijectively.

A particular, multiple context-free grammar (MCFG) for pk-structures is introduced in
\cite{Rivas:99}. This grammar employs a vector of nonterminal symbols referencing a
substructure with a gap.
The time and space complexity for Boltzmann sampling pk-structures with this grammar are
$O(n^4)$ and $O(n^2)$, respectively.
The grammar furthermore implies a folding algorithm having relatively high computational
cost \cite{Rivas:99}: $O(n^6)$ time and $O(n^4)$ space complexity, using the thermodynamic model.
Restricting to certain types of pk-structures the time complexity can be reduced to $O(n^4)$ time
complexity \cite{Dirks:03}. In any case, specifying the output space of the MCFG, i.e.~what types
of structures the grammar actually generates is subtle, see \cite{Rivas:00} for details.

The passage from MCFG to CFG has implications for time and space complexity:
an MCFG contains rules of the type $A\to B_1C_1B_2C_2$ where
$\langle B_1, B_2\rangle$, $\langle C_1, C_2\rangle$ are vectors of nonterminals.
It requires a minimum of three loop-indices to execute these rules in the sampling or folding
process. 
In contrast, the production rule in any CFG being in Chomsky normal form can be written as
$A\to BC$, which requires only one loop index, reducing the time complexity by $O(n^2)$.
As for space complexity this passage obsoletes using vectors, reducing the space complexity
for $O(n)$ and $O(n^2)$ for sampling and folding, respectively.

The key idea facilitating this passage lies in mapping the information of
crossing arcs into the labels of noncrossing arcs, i.e.~passing from pk-structures to
$\lambda$-structures.
Accordingly, the bijection of eq.~\ref{E:bi} is the crucial point. The fact that the bijection
is genuinely topology-based allows {\it RNAFeatures}$^*$ to extend the classic CFG of
secondary structures, generating any pk-structure.

In the following we present the expectation and variance of some important parameters 
related to the structural features of pk-structures generated by the induced SCFG
of {\it RNAFeatures}$^*$. These parameters are also used to analyze Boltzmann sampled
secondary structures in comparison with native structures in \cite{Nebel:11}.  
Our SCFG is trained via $90$ single stranded tRNA sequence-structure
pairs, taken from the Nucleic Acid Database (NDB) \cite{NDB:1, NDB:2}, considering only Watson-Crick and wobble
base pairs. In this training set, we find $16$ structures having genus $0$, $73$ of genus $1$ and $1$
structure of genus $2$. The average sequence length is $75.21$ nucleotides.
We analyze the following five random variables:
$bp$, the number of base pairs,
$st_n$, the number of stacks,
$st_{\ell}$, the average length of a stack,
$hp_n$, the number of hairpin and finally
$hp_{\ell}$, the average length of a hairpin loop.
Here a stack of length $k$ is the maximal sequence of parallel arcs $\{(i,j), (i+1, j-1), \ldots, (i+k-1, j-k+1)\}$.
A hairpin loop of length $j-i-1$ is a substructure, which consists of a base pair $(i,j)$ and
the unpaired bases $i+1, \ldots, j-1$.

The statistics is obtained by Boltzmann sampling of $10^5$ structures over $76$ nucleotides via the
SCFG. For reference purposes we furthermore sample $10^5$ random structures over $76$ nucleotides, having
fixed genus $0$ and $1$, respectively, with uniform probability \cite{Huang:13}. We summarize in
Table~\ref{T:stat} the respective means and variances of the above five random variables.

The sampling generates $264$, $9610$, $126$ structures having genus $0$, $1$ and $2$, reflecting the
large quantity of genus $1$ structures in the training set. By construction our SCFG does not recognize
contributions of unpaired bases in hairpin loops, bulges and interior loops, which explains the systematic
deviation in the lengths of hairpin-loops. Table~\ref{T:stat} shows that the Boltzmann sampled
structures are not random structures of either genus zero or one.

%Folding
The implication of the CFG to the folding of pk-structures is work in progress.
Though we can fold a $\lambda$-structure over a given sequence $\theta$ in $O(n^3)$ time,
the bijection producing the corresponding pseudoknotted structure does systematically alter
the underlying sequence. In particular, in case of pk-structures of genus one, this consists
in the transposition of two subsequent intervals and is thus manageable.

\section{ Acknowledgments.}
%Special thanks to Michael Waterman and Peter Stadler for their input on this manuscript. 
We gratefully acknowledge the help of Kevin Shinpaugh and the computational support team at VBI.
Many thanks to Christopher L.~Barrett and Henning Mortveit for discussions.

\bibliographystyle{elsarticle-num}

\newpage

\begin{figure}[t]
\includegraphics[width=0.9\columnwidth]{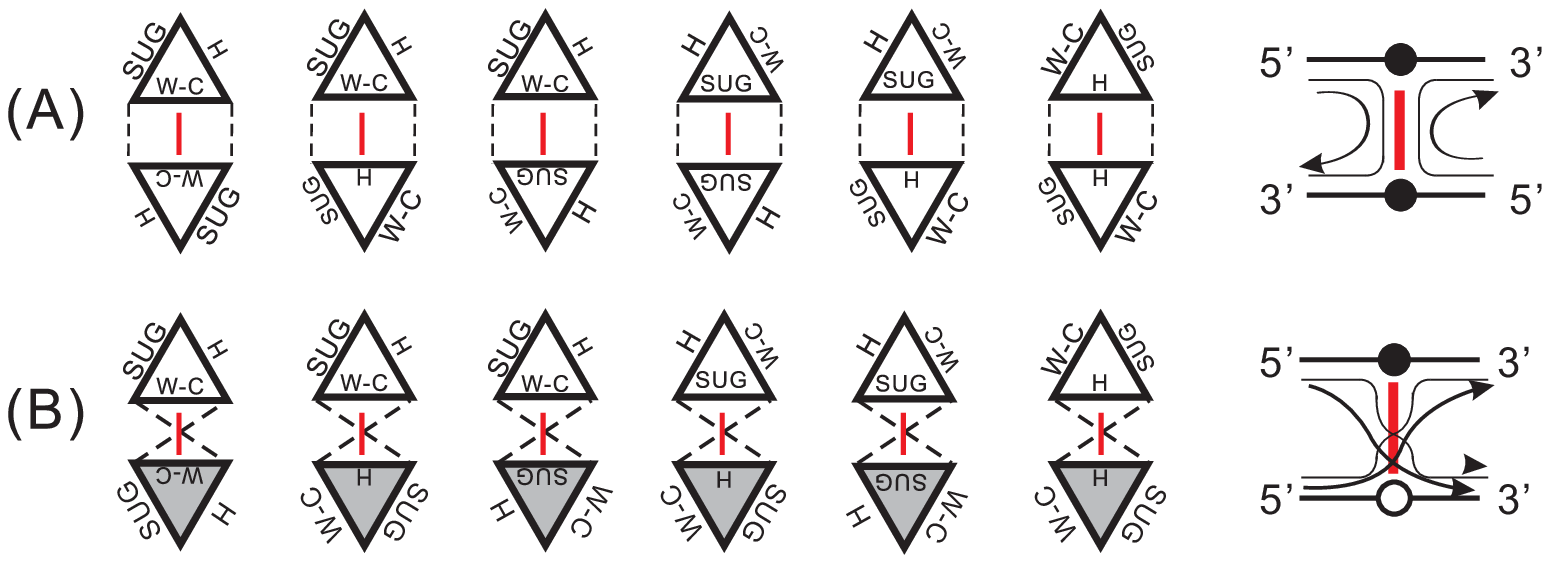}
\caption{\small 
Natural classification of RNA base pairs \cite{Leontis:01} by faces of simplices. (A) Cis-base pairs formed by 
parallel triangles inducing untwisted ribbons. (B) Trans-base pairs formed by anti-parallel
triangles inducing twisted ribbons. 
}
\label{F:basepair}
\end{figure}

\begin{figure}[t]
\includegraphics[width=\columnwidth]{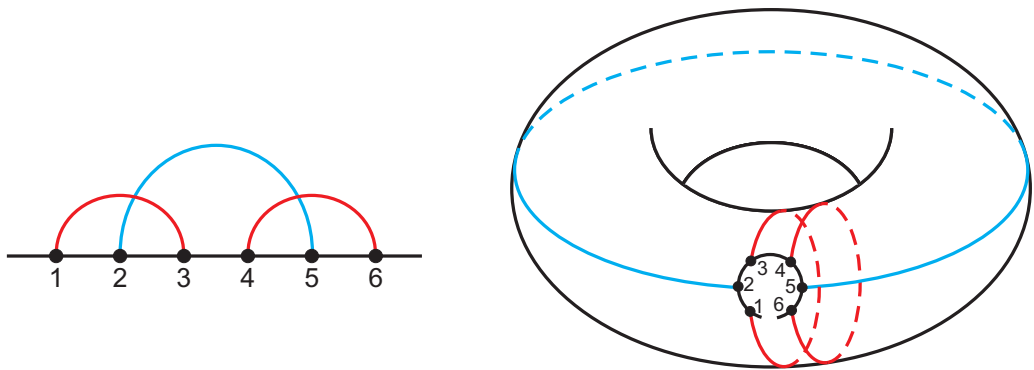}
\caption{\small 
  Crossings depend on the surface the structure is drawn on:
  here the diagram exhibits crossing arcs when drawn in the plane,
  on the torus however, the diagram can be drawn cross-free. }
\label{F:pse}
\end{figure}

\begin{figure}[t]
\begin{center}
\includegraphics[width=0.9\columnwidth]{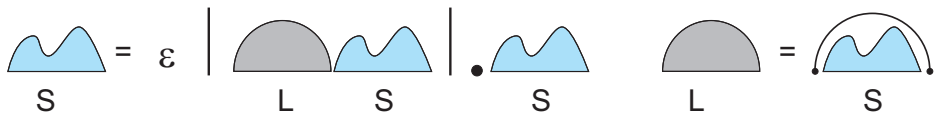}
\end{center}
\caption{\small Illustration of the two production rules of the CFG of RNA secondary structures:
  decomposition (left) and arc removal (right).
}
\label{F:grammar}
\end{figure}

\begin{figure}[t]
  \includegraphics[width=\columnwidth]{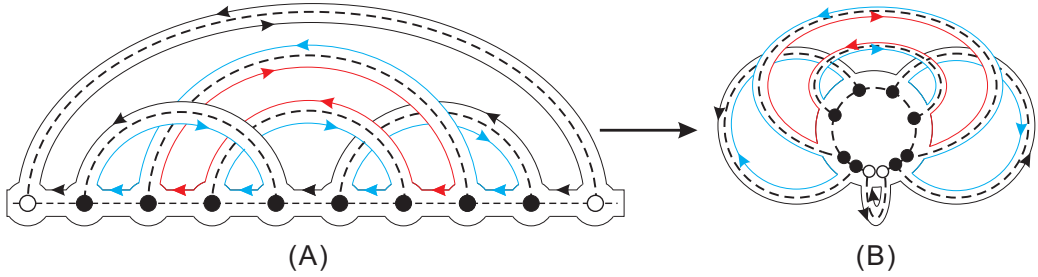}
\caption{\small (A): Inflation of edges and vertices into ribbons and disks.
  Here we have four boundary components traversed in counter-clockwise
  orientation traversing the sides of any ribbon in opposite
  directions. (B): Collapsing the backbone of (A) into a single disc does
  not affect genus. Here we have $g=1$.
} \label{F:fat}
\end{figure}

\begin{figure}[t]
\begin{center}
  \includegraphics[width=\columnwidth]{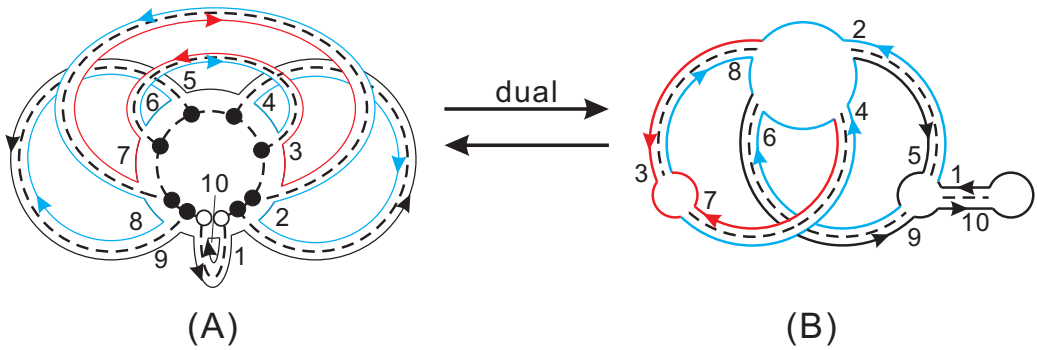}
\end{center}
\caption{\small (A): a fatgraph represented by permutations. We have $\sigma=(1,2,3,4,5,,6,7,8,9,10)$,
  the unique vertex obtained by collapsing the backbone, $\alpha = (1,10)(2,5)(3,8)(4,7)(6,9)$ the
  fixed-point free involution representing the ribbons, and $\gamma=(1,5,9)(2,8,6,4)(3,7)(10)$. 
(B): the Poincar\'{e} dual interchanging boundary components by vertices, hence producing a unicellular map. 
}
 \label{F:dual}
\end{figure}

\begin{figure}[t]
\includegraphics[width=\columnwidth]{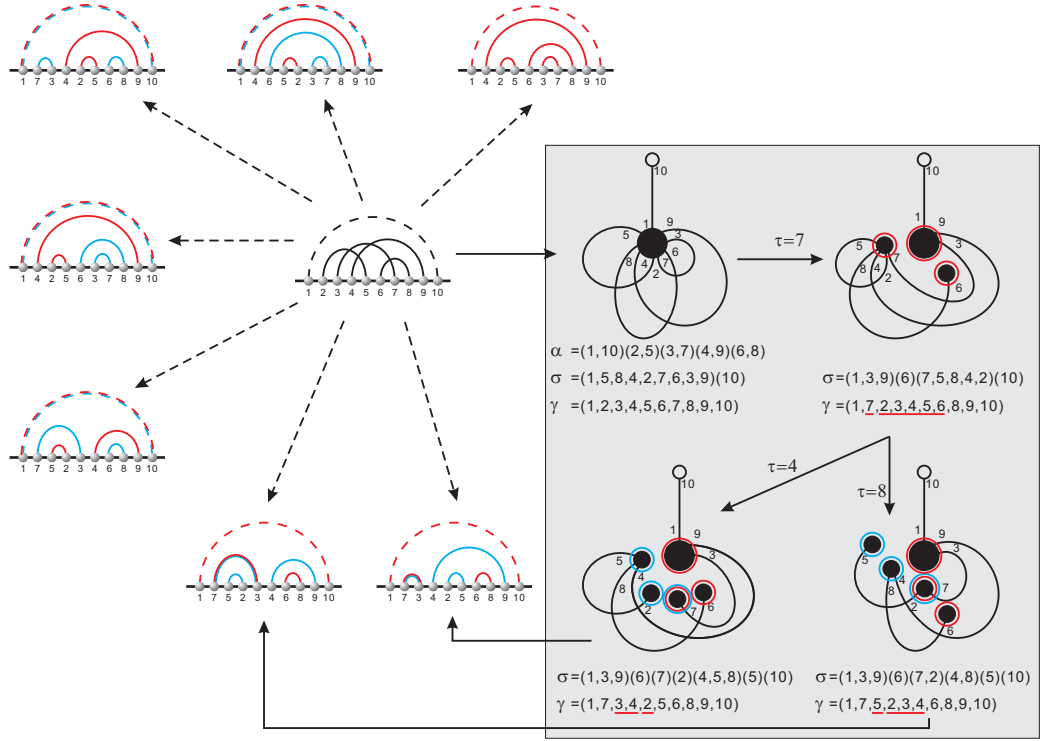}
\caption{\small All blueprints of a matching of genus $2$ as 
  well as all of its associated $\lambda$-matchings induced by all
  possible blueprints. In the insert, we show how the bottom two
  $\lambda$-matchings are induced:
  red circles reference the first
  and blue circles the second slicing, respectively.  
}
\label{F:path}
\end{figure}

%%%%%%%%%%%
%%%%%%%%%%%%%%%%example of labeled sec%%%%%%%%%%%%%%%
%%%%%%%%%%%
\begin{figure*}[t]
\begin{center}
\includegraphics[width=\columnwidth]{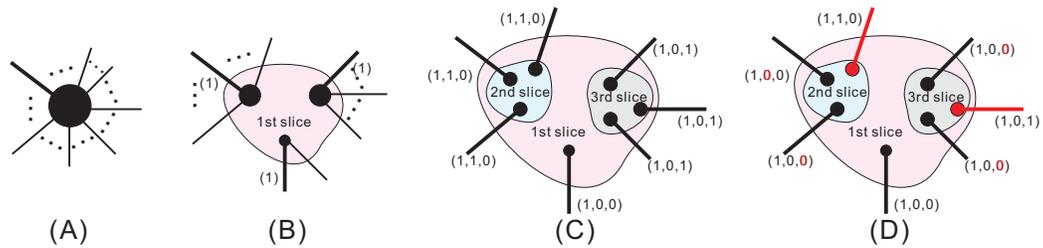}
\end{center}
\caption{ The idea of the proof of Lemma~\ref{L:welldefined}: the labels record the
  successive slicings and are finally reduced: only transitional vertices carry
  the information about previous slicing events.
}
\label{F:remove_label}
\end{figure*}
%%%%%%%%%%%
%%%%%%%%%%%%%%%%%%%%%%%%%%%%%%%%%%%%
%%%%%%%%%%%

%%%
%%%%%%%%%%%%%%%%%%%%%%%%%%%%%%%%%%%%%%%%%%%%%%%%%%%%%%%%%%%%%%%%%%%%%%%%%%%%%%%%%%%%%%%%%%%%%
%%%
\begin{figure*}[ht]
\begin{center}
  \includegraphics[width=\columnwidth]{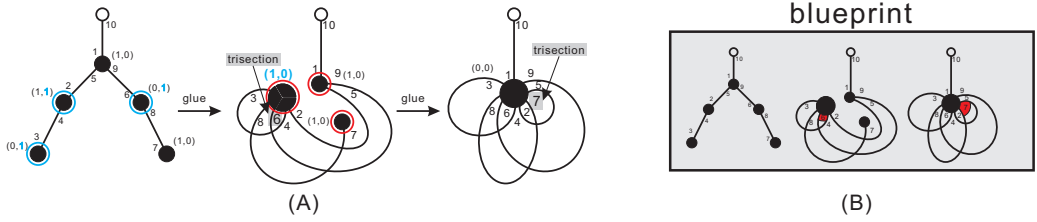}
\end{center}
\caption{Reconstructing the blueprint from a $\lambda$-tree: we consider the set of all vertices
  such that $\sigma_v|_2=1$ (blue circles) and glue and relabel the
  vertices as in the proof of Proposition\ref{P:bi}.
 This generates the unicellular map $\mathfrak{m}_1$ with the glued vertex
  labeled by $(1,0)$, carrying the distinguished trisection $8$.
  Iteration of this process, produces $\mathfrak{m}_2$ together with a distinguished
  trisection. We have thus constructed from the $\lambda$-tree the blueprint
  $((\mathfrak{m}_{2}, \tau), (\mathfrak{m}_{1}, \tau_1,), (\mathfrak{m}_0,\varnothing))$.
} \label{F:blueprint}
\end{figure*}
%%%
%%%%%%%%%%%%%%%%%%%%%%%%%%%%%%%%%%%%%%%%%%%%%%%%%%%%%%%%%%%%%%%%%%%%%%%%%%%%%%%%%%%%%%%%%%%%%
%%%

\begin{figure}[t]
  \includegraphics[width=0.9\columnwidth]{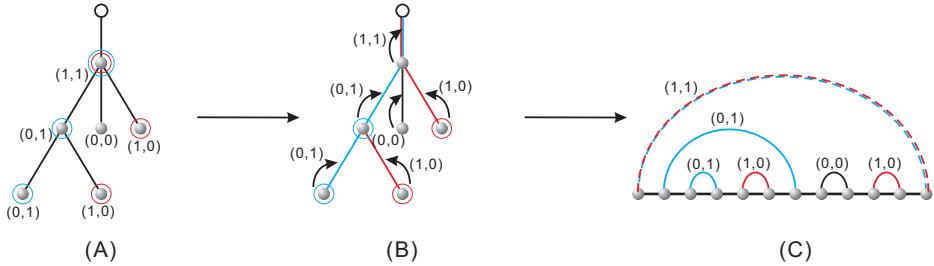}
\caption{\small From a $\lambda$-tree with labeled vertices to a $\lambda$-matching with labeled arcs.
  (A) a $\lambda$-tree with labeled vertices. (B) shifting the labels from vertices 
  to edges. (C) a $\lambda$-matching with labeled arcs obtained from (B) by taking the
  Poincar\'{e} dual.
}
\label{F:vtoe}
\end{figure}

\begin{figure}[t]
  \includegraphics[width=0.9\columnwidth]{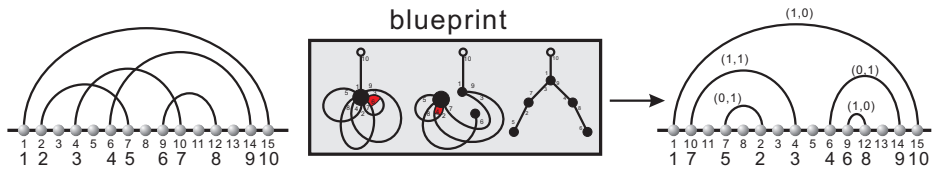}
\caption{\small A pk-structure together with a blueprint maps to a unique $\lambda$-structure. }
\label{F:comm1}
\end{figure}

\begin{figure}[t]
\includegraphics[width=\columnwidth]{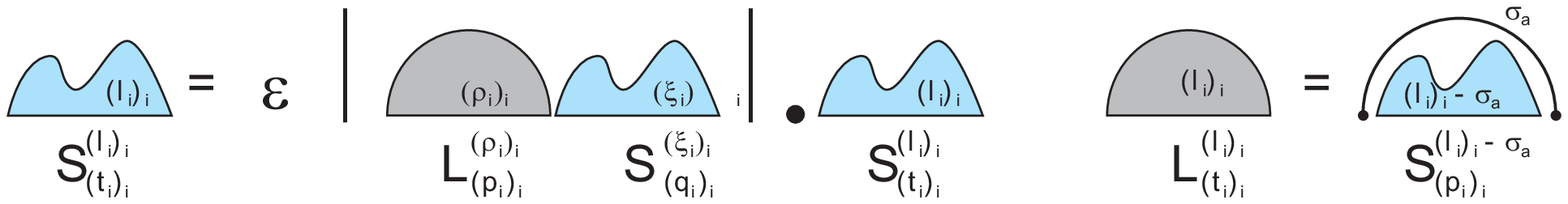}
\caption{\small Decomposition: any transitional arc is placed in the left-nonterminal as
  long as the latter contains any arcs $a$ with the property $\sigma_a|_i=1$.  
}
\label{F:grammar1}
\end{figure}

%%%%%%%%%%%
%%%%%%%%%%%%%%%%example of labeled sec%%%%%%%%%%%%%%%
%%%%%%%%%%%
\begin{figure*}[t]
\begin{center}
\includegraphics[width=\columnwidth]{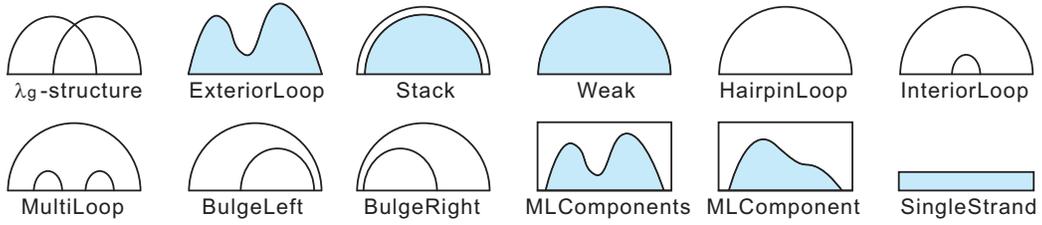}
\end{center}
\caption{ The nonterminal symbols of {\it RNAFeatures*}. }
\label{F:nonterminals}
\end{figure*}
%%%%%%%%%%%
%%%%%%%%%%%%%%%%%%%%%%%%%%%%%%%%%%%%
%%%%%%%%%%%

%%%%%%%%%%%
%%%%%%%%%%%%%%%%example of labeled sec%%%%%%%%%%%%%%%
%%%%%%%%%%%
\begin{figure*}[t]
\begin{center}
\includegraphics[width=0.9\columnwidth]{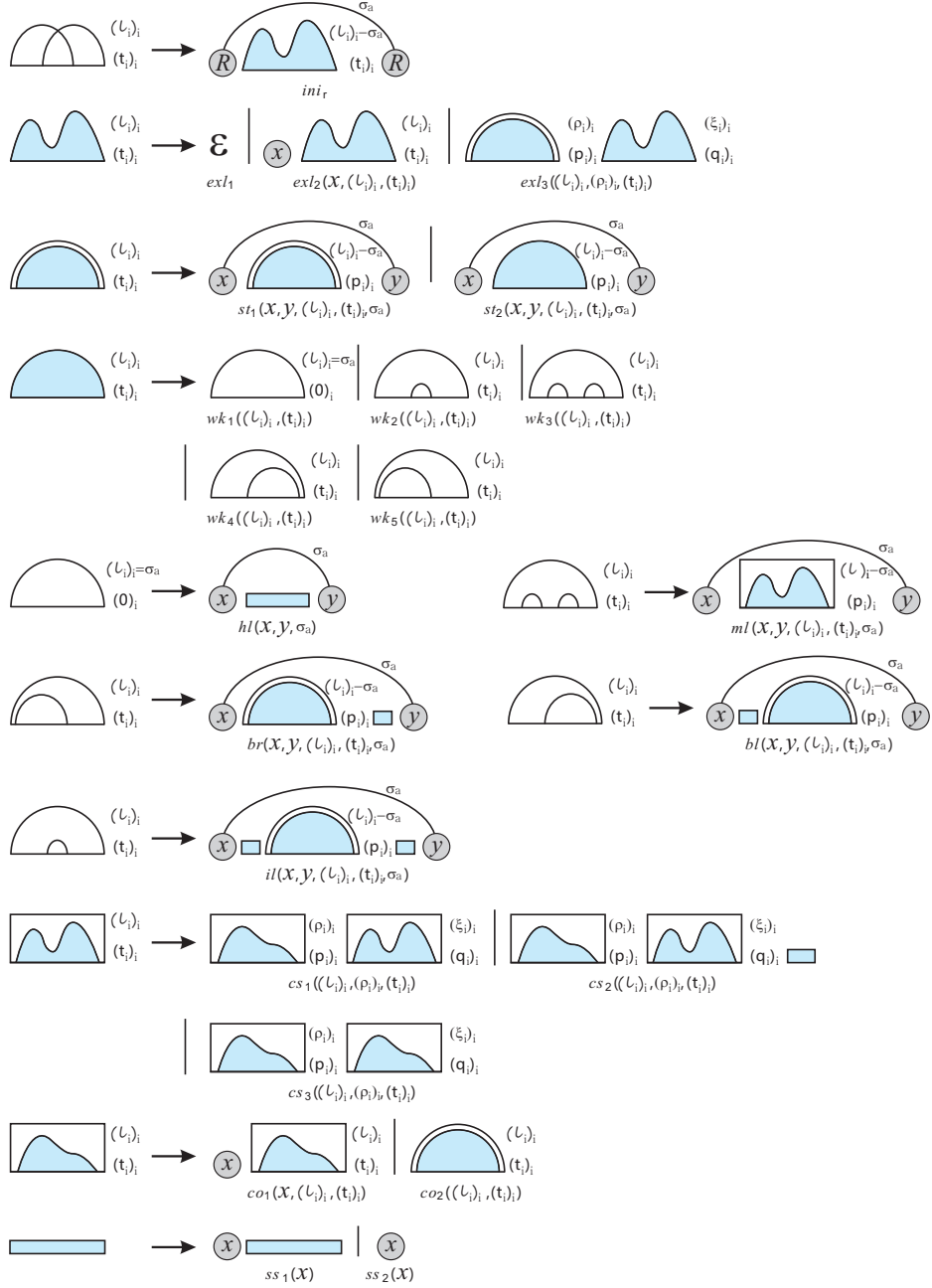}
\end{center}
\caption{ The production rules of {\it RNAFeatures*}. }
\label{F:rules}
\end{figure*}
%%%%%%%%%%%
%%%%%%%%%%%%%%%%%%%%%%%%%%%%%%%%%%%%
%%%%%%%%%%%

\begin{table}[t]
\caption{\small Expectation and variance of important statistic variables related to the RNA structural motifs.}
\label{T:stat} 
\begin{tabular}{|c|c|c|c|c|c|}
\hline
& \small Input Set & \small Boltzmann & \small Uniform, $g=1$ & \small Uniform $g=0$ \\
$bp$ & $21.31$ & $21.06$ & $25.58$ & $25.05$ \\ \hline
$E(st_n)$ & $5.17$ & $5.36$ & $22.94$ & $22.25$ \\ \hline
$Var(st_n)$ & $0.56$ & $0.91$ & $3.98$ & $4.25$ \\ \hline
$E(st_\ell)$ & $3.77$ & $3.65$ &$1.03$ & $1.04$ \\ \hline
$Var(st_\ell)$ & $4.68$ &$ 9.92$ & $0.97$ & $0.97$ \\ \hline
$E(hp_n)$ & $1.29$ & $1.18$ & $11.81$ & $13.04$ \\ \hline
$Var(hp_n)$ & $0.50$ & $0.21$ & $4.27$ & $4.29$ \\ \hline
$E(hp_\ell)$ & $7.16$ & $5.55$ & $0.47$ & $0.51$ \\ \hline
$Var(hp_\ell)$ & $6.85$ & $13.38$  & $0.70$ & $0.74$ \\ \hline
\end{tabular}
\end{table}

%% The Appendices part is started with the command \appendix;
%% appendix sections are then done as normal sections
%% \appendix

%% \section{}
%% \label{}

%% If you have bibdatabase file and want bibtex to generate the
%% bibitems, please use
%%
%%  \bibliographystyle{elsarticle-harv} 
%%  \bibliography{<your bibdatabase>}

%% else use the following coding to input the bibitems directly in the
%% TeX file.

%\begin{thebibliography}

%% \bibitem[Author(year)]{label}
%% Text of bibliographic item

%\bibitem[ ()]{}

%\end{thebibliography}
\end{document}